\documentclass[
superscriptaddress,
 amsmath,amssymb,
 aps, physrev,
pra,
]{revtex4-2}

\usepackage{graphicx}
\usepackage{bm}
\usepackage{amsmath}
\usepackage{physics}
\usepackage{amssymb}
\usepackage{amsfonts}
\usepackage{amsthm}
\usepackage{bbm}
\usepackage{braket}
\usepackage[normalem]{ulem}
\usepackage{wrapfig}
\usepackage{tikz}
\usepackage{dsfont}
\usepackage{comment}
\usepackage{thmtools,thm-restate}
\usepackage[export]{adjustbox}
\usepackage{mathtools}
\usepackage[noend]{algcompatible}
\usepackage{algorithm}
\usepackage{pdfpages}
\usepackage{outlines}
\usepackage{xcolor}
\usepackage{soul}
\usetikzlibrary{arrows.meta}
\definecolor{darkblue}{rgb}{0,0,0.5}
\usepackage{hyperref}
\hypersetup{
colorlinks=true,
linkcolor=black,
filecolor=blue,
citecolor=darkblue,  
urlcolor=black,
}

\usepackage[strict]{changepage}

\usepackage{setspace}

\newtheorem*{observation}{Observation 1}

\newtheorem*{remark}{Remark}
\newtheorem{theorem}{Theorem}
\newtheorem{definition}{Definition}
\newtheorem{corollary}{Corollary}
\newtheorem{lemma}{Lemma}
\newtheorem*{theorem*}{Theorem}

\newcommand{\be}{\begin{equation}}
\newcommand{\ee}{\end{equation}}

\DeclareMathAlphabet{\mathbbold}{U}{bbold}{m}{n}

\makeatletter

\newcommand{\FigQLtoQM}
{
\begin{figure}[htb!]
    \centering
    \includegraphics[width=\textwidth]{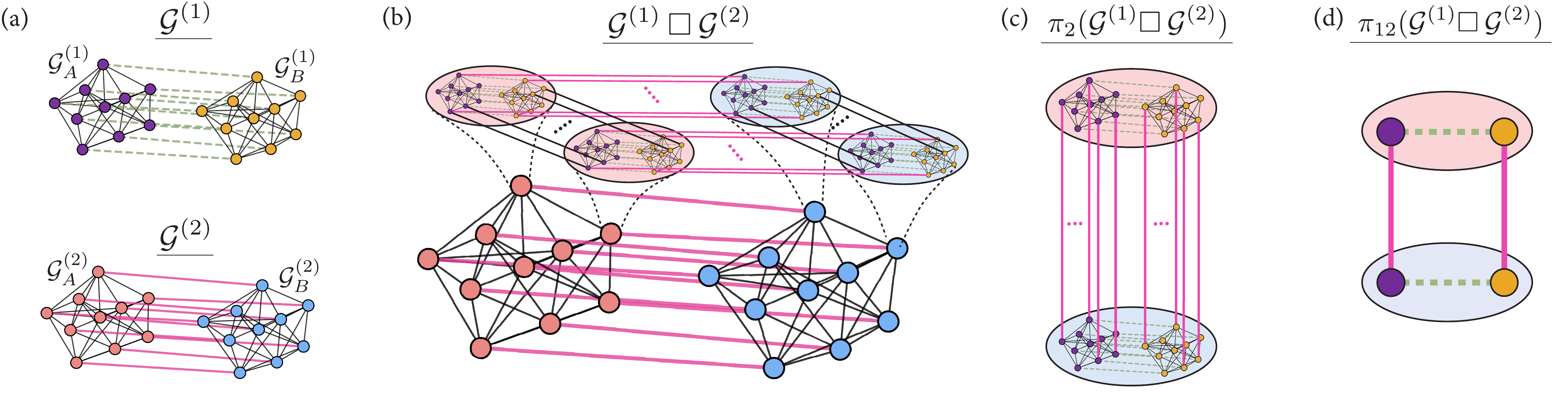}
    \caption {
    \textbf{(a)} Graph $\mathcal{G}^{(1)}$ with subgraphs $\mathcal{G}^{(1)}_{A}$ and $\mathcal{G}^{(1)}_{B}$ shows the construction of a QL-bit. The subgraphs are $k$-regular, with $k = 5$, while the coupling between them is $\ell$-regular, $\ell$ being equal to 1. QL-bit graph $\mathcal{G}^{(2)}$ is also constructed in a similar manner. \textbf{(b)} Cartesian product $\mathcal{G} = \mathcal{G}^{(1)} \Box \hspace{2pt}\mathcal{G}^{(2)}$ of QL-bits in (a). For each node in $\mathcal{G}^{(2)}$, there is a copy of $\mathcal{G}^{(1)}$. \textbf{(c)} Quotient graph formed by taking the equitable partition generated by holding constant the subgraph indices of second graph $\mathcal{G}^{(2)}$ in the Cartesian product.
    \textbf{(d)} The minimal representation of the QL-bit product, subgraphs of each QL-bit are reduced to nodes and the coupling $\ell$ for between each subgraph persists.
    }
    \label{fig: ql to qm}
\end{figure}
}

\newcommand{\FigQLProds}
{
\begin{figure}[htb!]
    \centering
    \includegraphics[width=\textwidth]{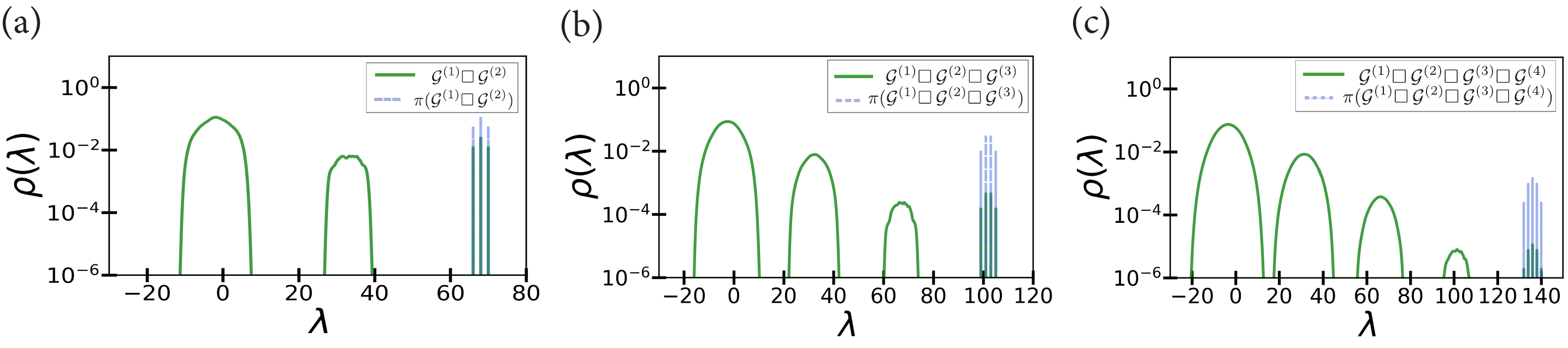}
    \caption{
    Spectra of full Cartesian products and quotient graph when coupling \textbf{(a)} 2, \textbf{(b)} 3, and \textbf{(c)} 4 QL-bits. Each QL-bit is formed from subgraphs of size $|\mathcal{G}^{(q)}_{A, B}| =  40$, $k = 34$ and $\ell=1$. The spectra for both the full Cartesian product and the quotient graph, corresponding to the minimal representation, were generated by the convolution method described in Eq(2.15) of~\cite{amati_quantum_2025}, with the discrete set of incoherent states smoothed by a Gaussian filter.
    }
    \label{fig: quantumlike products}
\end{figure}
}

\newcommand{\FigEffConstruct}
{
\begin{figure}[htb!]
    \centering
    \includegraphics[width=0.8\textwidth]{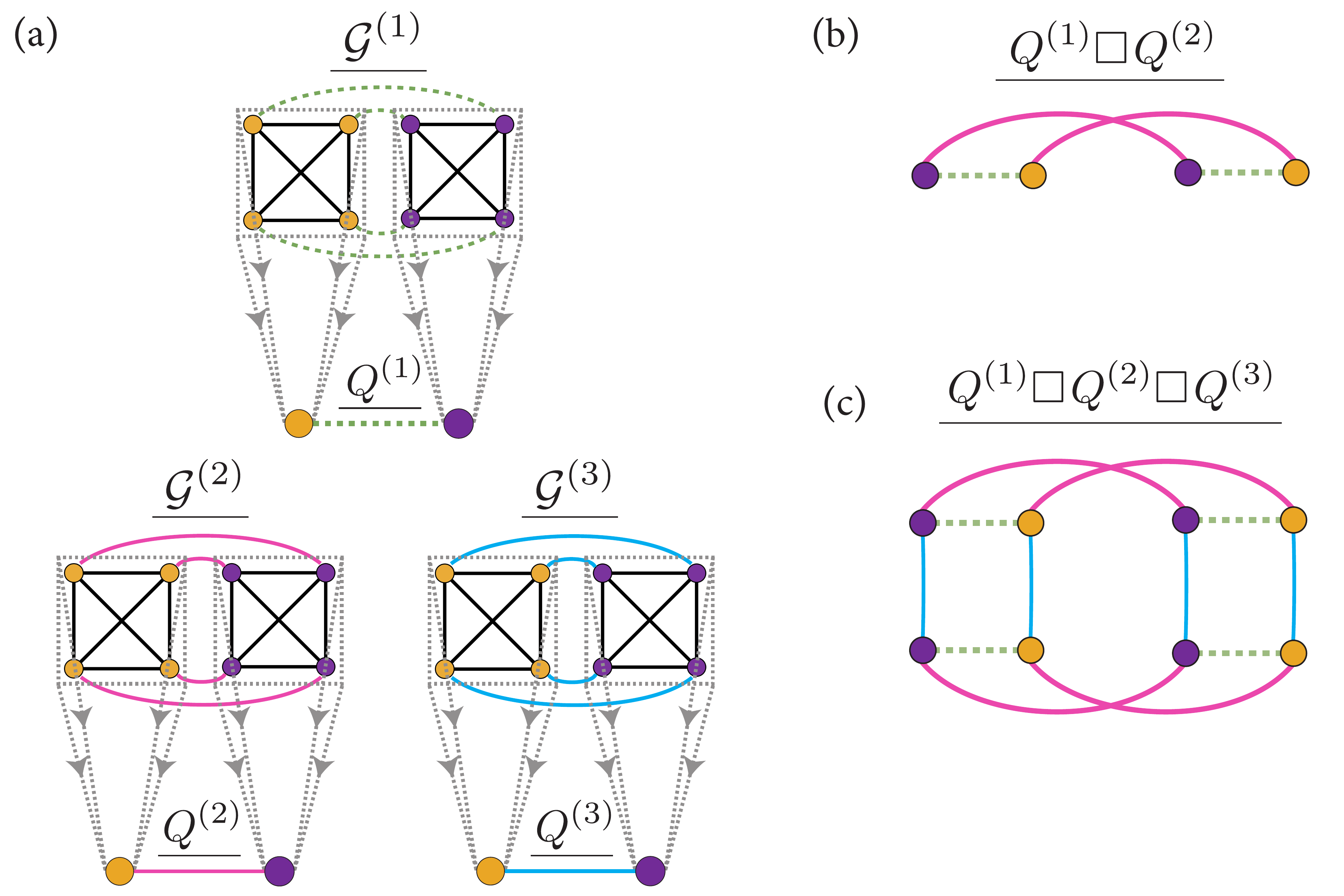}
    \caption{
    Method to form quantum-like products without first generating the Cartesian product: \textbf{(a)} QL-bits $\mathcal{G}^{(1)}$, $\mathcal{G}^{(2)}$ and $\mathcal{G}^{(3)}$ are contracted to their minimal representation quotient graph. Their respective subgraphs are reduced to single vertices (shown using dotted gray arrows) while preserving the coupling $\ell$ between nodes. This allows for the construction of the quotient graph without first forming the full Cartesian product from QL-bits. \textbf{(b)},\textbf{(c)} Direct construction of quotient graphs, $\pi(\mathcal{G}^{(1)} \Box \hspace{2pt} \mathcal{G}^{(2)}) = Q^{(1)} \Box \hspace{1pt} Q^{(2)}$ and $\pi(\mathcal{G}^{(1)} \Box \hspace{2pt} \mathcal{G}^{(2)} \Box \hspace{2pt} \mathcal{G}^{(3)}) = Q^{(1)} \Box \hspace{1pt} Q^{(2)} \Box \hspace{1pt} Q^{(3)}$. The loops generated in taking equitable partition in this construction of composite QL systems have been removed to keep the construction clear. Note that the graphs here are only representational, and that the QL-bit graphs require far more number of nodes per subgraph to have the quantum-like properties.
    }
    \label{fig: efficient construction}
\end{figure}
}

\newcommand{\FigBaseGraph}
{
\begin{figure}[htb!]
    \centering
    \includegraphics[width=0.95\textwidth]{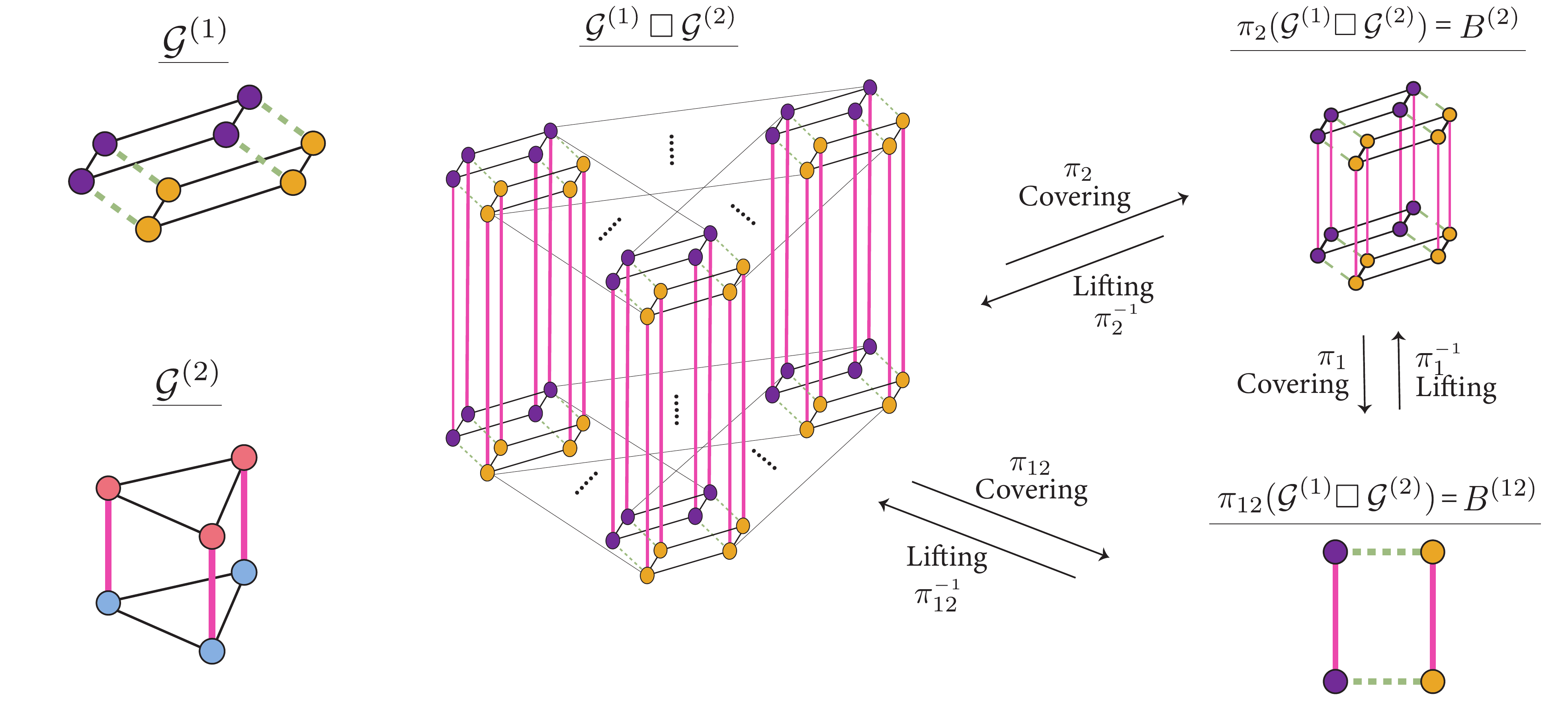}
    \caption{
    Equivalence of the base graph and quotient graph. The Cartesian product of QL-bits $\mathcal{G}^{(1)}$ and $\mathcal{G}^{(2)}$ is reduced to the quotient graph via an equitable partition $\pi$. This reduction is equivalent to a graph fibration that functions as a covering projection over the uniform subgraphs. The inverse operation of the equitable partition is realized through the unique lifting of the base graph's adjacency structure, exactly reconstructing the full graph along the fibers defined by $\pi^{-1}$.
    }
    \label{fig: proof}
\end{figure}
}

\makeatletter
\AtBeginDocument{\let\LS@rot\@undefined}
\makeatother

\allowdisplaybreaks

\begin{document}


\title{Minimal representations of topology-preserving quantum-like states}%

\author{William Horvat}
\thanks{These authors contributed equally to this work.}
\affiliation{College of Letters and Science, University of California, Los Angeles, CA 90095, USA}

\author{Debadrita Saha}
\thanks{These authors contributed equally to this work.}
\affiliation{Department of Chemistry, Princeton University, Princeton, NJ 08544, USA}

\author{Nicol\'as Gigena}
\affiliation{IFLP/CONICET - Departamento de F\'isica, Universidad Nacional de La Plata, C.C. 67, La Plata (1900), Argentina}

\author{Prineha Narang}
\affiliation{College of Letters and Science, University of California, Los Angeles, CA 90095, USA}

\author{Gregory D. Scholes}
\affiliation{Department of Chemistry, Princeton University, Princeton, NJ 08544, USA}
\date{\today}

\begin{abstract}
We provide an equitable partition that gives an exact, minimal representation  for the graph Cartesian product formed from quantum-like bits that preserves the relevant spectral and topological properties. We show that this result follows from the fact that the operations of taking the Cartesian product of graphs and constructing equitable partition of the graphs commute. Numerical simulations illustrate the preserved emergent eigenstates in the reduced structures. Additionally, we provide a construction of the minimal structure without passing through the Cartesian product. Finally, we frame quantum-like structures in the language of topology and fibrations. 
\end{abstract}

\maketitle




\section{Introduction}\label{sec:introduction}
Complex classical networks arise naturally in a wide range of physical, biological, and computational systems, where interactions between components can give rise to collective emergent behavior. Graphs provide an abstract way to represent the correlations and study the emergent behavior arising from synchronization in these complex networks. Ref.~\cite{scholes_quantumlike_2024} establishes the formalism for a quantum-like (QL) bit within the framework of a classically synchronizing network. The study shows that connected subgraphs with expander properties have emergent eigenstates isomorphic to the Hilbert space of a qubit. These emergent states are isolated from the incoherent state eigenspectrum of the system, making them robust to decoherence and suitable for encoding and processing quantum information. 
Building systems from multiple QL-bits is achieved by employing Cartesian products of QL-bit graphs. However, this comes at a computational cost that scales exponentially with the number of QL-bits. The Cartesian product of graphs $\mathcal{G}^{(1)}$ and $\mathcal{G}^{(2)}$ implies that we create a copy of graph $\mathcal{G}^{(1)}$ at each node of graph $\mathcal{G}^{(2)}$ therefore leading to a resource scaling of $N^{N_{QL}}$, where $N_{QL}$ is the number of QL-bits and $N$ is the total number of nodes in a QL-bit,  $N = N_{\mathcal{G}_A} + N_{\mathcal{G}_B}$~\cite{hammack_handbook_2011}. Since these graphs represent some underlying classical network, the number of physical resources required increases exponentially with the number of QL-bits used to build the Cartesian product graph.

Here  we address the problem of exponential scaling in constructing these composite QL-systems. For that, we employ graph partitioning methods which provide a systematic framework for decomposing large graphs into smaller structured components based on connectivity or symmetry~\cite{drapeau_complete_2024}. Graph partitioning methods find application across a broad spectrum of scientific and engineering fields that require large, complex networks to be decomposed into more tractable components while preserving essential structural properties. 
In several scientific computing applications, computation can be conveniently decomposed in the language of graphs~\cite{pothen_1997_graph}. These include finite element and finite difference calculations, molecular dynamics simulations~\cite{hristo_2019_qmoldyn}, particle-in-cell codes~\cite{birdsall_particle--cell_1991}, and a variety of other computations. Other prominent applications include the decomposition of data structures for parallel computation~\cite{hendrickson_graph_2000}, devising efficient circuit layouts for VLSI design~\cite{hagen_1992_circuit,kahng_2011_vlsi}, the ordering of sparse matrix computations~\cite{lanczos_solution_1952}, and in the study of biological and social networks, it reveals community structure and functional modularity that would otherwise be obscured by the global complexity of the network~\cite{de_domenico_mathematical_2013,bunimovich_finding_2019,boguna_network_2021}. These applications and many more have led to several graph partitioning techniques and algorithms~\cite{hendrickson_multilevel_1995}.

Here we employ a special class of graph partitioning, called equitable partitions, where every vertex in a given cell has the same number of neighbors in each other cell. This regularity condition ties directly to the spectral properties of the graph through the quotient matrix, a smaller matrix that summarizes the inter-cell edge structure. A key property is that the eigenvalues of the quotient matrix are always a subset of the eigenvalues of the full adjacency matrix, making equitable partitions a useful tool for spectral analysis~\cite{drapeau_complete_2024}. Equitable partitions arise naturally in the study of cluster synchronization in oscillator networks~\cite{aguiar_syncep_2018,schaub_cluster_2016}, and in studies of continuous random walks~\cite{ge_pst_2011,bachman_pst_2012,cancado_quotientgraphs_2024} on graphs, where perfect state transfers in a full graph is possible if and only if it is shown to be possible in its quotient graph for an equitable partition, $\pi$.

A complementary perspective is offered by the theory of graph fibrations which are structure-preserving maps between directed graphs satisfying a local lifting condition on edges, analogous to fibrations in algebraic topology. This notion has appeared, often implicitly, across several areas, including spectral graph theory, distributed computing, symbolic dynamics, graph neural networks, and category theory~\cite{boldi_fibrations_2002}. In the context of directed graphs, Boldi and Vigna formalized graph fibrations as a categorical generalization of graph coverings~\cite{boldi_fibrations_2002}. In topology, a fiber bundle consists of a total space, a base space, and a projection map, such that the total space is locally homeomorphic to a product of the base space with a typical fiber; this local triviality ensures that neighborhoods in the base lift consistently to neighborhoods in the total space. Graph fibrations and coverings provide a natural discrete analog of this structure~\cite{leighton_finite_1982,amit_random_2002}, and random lifts of graphs, where fibers are assigned by random permutations have been shown to inherit spectral properties of the base graph, including edge expansion and optimal spectral gap~\cite{amit_randomlifts_2006,bilu_lifts_2006,marcus_interlacing_2013}. Here, the total graph plays the role of the total space, the base graph corresponds to the base space, and the preimage of each base vertex defines a discrete fiber. The fibration condition requires that edges terminating at a vertex in the base graph admit unique lifts to edges terminating at any vertex in its fiber, enforcing a consistent local structure across fibers. From this perspective, equitable partitions can be interpreted as fiber decompositions, capturing local symmetries and structural invariants while reducing the complexity of the total graph. 

In the present work, we leverage the symmetries enforced by the Cartesian product construction and equitable graph partition to construct a minimal representation that maintains the emergent states,  eigenvalues, and relevant topology of the full Cartesian product of QL-bit graphs.

The paper is organized as follows: In Section~\ref{sec:background} we present the relevant notions of graph partitioning. In Section~\ref{sec:minrep} we state the papers main theorem on minimal representations of QL products, along with a discussion on resource scaling and efficient multi-QL-bit systems construction. In Section~\ref{sec:fibrations} we offer a connection between QL-bit structure and graph topology. In Section~\ref{sec:conclusion}, we offer open-questions and directions to further develop quantum-like dynamics and physical implementation. In the Appendices~\ref{thm1_proof}, ~\ref{appendix:minrepcp}, ~\ref{lemma3_proof}, and ~\ref{lemma4_proof} we provide the proofs of the theorems and lemmas stated in the paper. 

%

\section{Preliminaries}\label{sec:background}
In this section we introduce the ingredients needed for the rest of the paper. Section~\ref{subsec:qlbits} reviews the construction of QL-bits from \textit{k}-regular expander graphs and their Cartesian products, mainly following Refs.~\cite{scholes_quantumlike_2024} and ~\cite{scholes_qlcartpdt_2025}. Section~\ref{subsec:equitable} defines equitable partitions and their associated quotient (divisor) graphs, and states the spectral relationship between a graph and its quotient that underlies the minimal representations developed in Section~\ref{sec:minrep}. Section~\ref{subsec:qlbits}, extending the construction to almost equitable partitions and establishing eigenvalue bounds for QL-bits with approximate regularity.

\subsection{Quantum-like resource from expander graphs}\label{subsec:qlbits}
 Synchronization in complex classical networks can give rise to emergent states that are spectrally isolated and robust to decoherence, as discussed in Section I. Ref.~\cite{scholes_quantumlike_2024} shows how the properties of expander graphs representing such networks can be used to construct two-level systems that serve as a quantum-like (QL) bit. In the following, we give a short overview of these results. Before proceeding, we first fix some notation: we use $\mathcal{G}$ to denote QL-bits and their graph representations, reserving $G$ for graphs discussed from the graph theory perspective.

We construct a single QL-bit as a network represented by a graph, $\mathcal{G}^{(q)}$, composed of two interacting subgraphs, denoted as $\mathcal{G}^{(q)}_{A}$ and $\mathcal{G}^{(q)}_{B}$~\cite{scholes_quantumlike_2024}. 
A graph $\mathcal{G} = (V, E)$ is defined by its vertex set $V(\mathcal{G})$ and edge set $E(\mathcal{G})$. Each subgraph, $\mathcal{G}^{(q)}_{A}$ and $\mathcal{G}^{(q)}_{B}$, is an undirected $k$-regular random graph, a prototypical example of an expander graph. They are interconnected such that every node in $\mathcal{G}^{(q)}_{A}$ is connected to $\ell$ nodes in $\mathcal{G}^{(q)}_{B}$. Although the graph representation is abstract, it can be made concrete by identifying the nodes in the graph as classical oscillators, and the edges between the nodes in the graph as couplings between those oscillators. 

The adjacency matrix for $\mathcal{G}$ representing a QL-bit, $\mathcal{R}^{(q)}$, takes the following block-structured form:
\begin{align}
    \mathcal{R}^{(q)} = 
    \begin{bmatrix}
        A_{A} & C \\
        C^{T} & A_{B}
    \end{bmatrix}
\end{align}

\noindent where the diagonal blocks, $A_{i}$, represents the adjacency matrix of the $k$-regular subgraphs, and the off-diagonal blocks $C$ represents the adjacency structure of the $\ell$-regular coupling between the vertex sets $V(\mathcal{G}^{(q)}_{A})$ and $V(\mathcal{G}^{(q)}_{B})$. The two largest eigenvalues of $\mathcal{R}^{(q)}$ at $k\pm \ell$ that are isolated from its incoherent state spectrum and their corresponding eigenstates are considered as the two states of the QL-bit. Although the class of expander graphs are valid structures for constructing QL-bits, we focus here on $k$-regular graphs. We also consider strict \textit{k}- and $\ell$-regularity here but earlier work shows that the emergent state is robust to structural, as well as energetic perturbations, in the network~\cite{scholes_large_2023,scholes_quantumlike_2024}. 

To form composite systems comprised of multiple QL-bits, we employ the Cartesian product of graphs~\cite{scholes_qlcartpdt_2025} or the graph box product. The Cartesian product of $N_{QL}$ such QL-bit graphs is given by $\mathcal{G} = \mathop{\Box}\limits_{q=1}^{N_{QL}} \mathcal{G}^{(q)}$, and the corresponding adjacency matrix, $\mathcal{R}$, can be expressed in terms of the individual adjacency matrices, $\mathcal{R}^{(q)}$, of the QL-bits as
\begin{align}
    \mathcal{R} = \sum_{q=1}^{N_{QL}} \left[ \bigotimes_{p=1}^{q-1} \mathbbold{1}_{2n} \right] \otimes R^{(q)} \otimes \left[ \bigotimes_{p' = q+1}^{N_{QL}} \mathbbold{1}_{2n} \right],
    \label{eqn:cartprodt}
\end{align}
where $n$ is the number of vertices in each subgraph of $\mathcal{G}^{(q)}$. The construction of the Cartesian product graph, $\mathcal{G} = \mathcal{G}^{(1)} \Box  \hspace{2pt} \mathcal{G}^{(2)}$, is shown in Figure.~\ref{fig: ql to qm}(b). The Cartesian product of individual QL-bits has been shown to map to the tensor product of Hilbert spaces corresponding to quantum bits~\cite{scholes_qlcartpdt_2025}. From this construction, the eigenspectrum of $\mathcal{R}$ has $2^{N_{QL}}$ emergent eigenvalues that can be expressed as the sums of the emergent eigenvalues, $\lambda_{i}^{(q)}$, of the constituent graphs~\cite{scholes_qlcartpdt_2025}. The spectrum of $\mathcal{G}$ will have all possible values given by $\sum\limits_{q=1}^{N_{QL}}\sum\limits_{i}\lambda_{i}^{(q)}$. The corresponding emergent eigenstates of these product graphs are tensor products of the emergent states of the constituent QL-bits.

Furthermore, to enable encoding and processing of information using these QL resources, Ref.~\cite{amati_quantum_2025} introduces QL analogs of standard quantum gates like the single qubit Hadamard and the Pauli matrices, as well as the two QL-bit controlled-NOT gate. Gate transformations are shown as unitary transformations on the graph adjacency matrices, therefore affecting the underlying physical network of oscillators.

\subsection{Equitable partitions of graphs}\label{subsec:equitable}

In this section, we define equitable graph partitions and introduce the quotient graph. We also discuss the relationship between the eigenspectrum of the full graph and that of its reduced representation, which allow us to establish the minimal representations of a composite system of QL-bits in Section~\ref{sec:minrep}. For the following we consider a graph $G=(V, E)$ with $V$ vertices and a set of $E$ edges defined by an adjacency matrix $A$. A graph partition is a division of the vertex set $V$ into disjoint non-empty subsets $\{C_1, C_2, \dots, C_k\}$, called cells, such that every vertex belongs to exactly one cell. 

\begin{definition}[Equitable Partition]
    The vertex set partition $\pi = \{C_1, C_2, \dots, C_k\}$ is called an equitable partition if for any two cells $C_i$ and $C_j$, every vertex $v \in C_i$ has exactly $d_{ij}$ neighbors in $C_j$.
\end{definition}

In terms of the adjacency matrix $A$, partitioning the vertices according to $\pi$ induces a block structure:
\begin{equation*}A = 
\begin{bmatrix}
A_{11} & \cdots & A_{1k} \\
\vdots & \ddots & \vdots \\
 A_{k1} & \cdots & A_{kk}
\end{bmatrix}
\end{equation*}
where each block (or submatrix) $A_{ij}$ describes the edges between the cells $C_{i}$ and $C_{j}$. The partition $\pi$ is equitable if and only if every block $A_{ij}$ has constant row sums, that is, every vertex in cell $C_{i}$ has exactly $d_{ij}$ neighbors in cell $C_{j}$, regardless of which vertex is chosen~\cite{godsil_compact_1997}. A note on notation, we let $\pi_i: G^{(i)}$ be the equitable partition for graph $i \in \{1, 2, \dots, n\}$, and the product of equitable partitions is $(\pi_i \times \pi_j) = \pi_{ij}$.

Associated with an equitable partition $\pi$ is the \textit{quotient graph} (or divisor graph), we will denote this as $\pi(G)$, $Q$, or $G_{\pi}$, breaking with convention on the first representation. The adjacency matrix of the quotient graph is the $k \times k$ \textit{divisor matrix} (or quotient matrix), $A_{\pi}$, where the entry $A_{\pi,{ij}} = d_{ij}$ represents the degree from a node in cell $C_i$ to the set of nodes in cell $C_j$. 

The spectral significance of the equitable partition is captured by the relationship between the characteristic polynomials of the original adjacency matrix $A$ and the divisor matrix $A_\pi$. Specifically, the characteristic polynomial of the divisor matrix, $\phi(A_\pi) = \det(\lambda I - A_\pi)$, divides the characteristic polynomial of the full graph, $\phi(G) = \det(\lambda I - A)$. Lemma~\ref{lemma: quotient adj} and the subsequent remarks detail the transformation from the full adjacency matrix to the divisor (quotient) matrix.

\begin{lemma} \label{lemma: quotient adj}
    \emph{(Quotient matrix from characteristic matrix)}. Let $\pi$ be an equitable partition of the graph $G$ with characteristic matrix $S$, $A$ the adjacency matrix of $G$, and $A_\pi$ the adjacency matrix of the $\pi(G)$~\cite{drapeau_complete_2024}. Then,
    \begin{gather*}
    AS = SA_{\pi} \\
    A_{\pi} = (S^T S)^{-1} S^T A S.
    \end{gather*}
\end{lemma}

We construct the characteristic matrix, $S$, in the manner outlined in~\cite{drapeau_complete_2024}. Given a graph $G$ on $|V(G)| = n$ vertices with a vertex set partition $\pi = \{C_1,\dots, C_k\}$, the characteristic matrix of $G$ with respect to $\pi$ is the matrix $S = [s_{ij}] \in \{0, 1\}^{|V(G)|\times k}$ where each column of $S$ represents a partition element, and each row represents a vertex. Entries $s_{ij}$ are,
$$s_{ij} = 
\begin{cases}
  1 & \text{if $i \in C_j$} \\
  0 & \text{otherwise}.
\end{cases}
$$

\noindent Since the columns of $S$ are orthogonal and each vertex belongs to exactly one cell, $S^T S$ is the nonsingular diagonal matrix $diag(|C_1|, |C_2|, \dots, |C_k|)$

Considering the eigenspectrum of $A$ and $A_{\pi}$, Lemma~\ref{lemma: eigenpairs} states the relationship of the eigenpairs of the two adjacency matrices. Particularly, every eigenpair of the divisor matrix lifts to an eigenpair of the full adjacency matrix~\cite{drapeau_complete_2024,brouwer_spectra_2012}. 

\begin{lemma} \label{lemma: eigenpairs}
    \emph{(Eigenpairs of the divisor matrix)}. If $\pi$ is an equitable partition of a graph G with an Hermitian adjacency matrix and ($\lambda$, \textbf{v}) is an eigenpair of the divisor matrix $A_{\pi}$, then ($\lambda$, S\textbf{v}) is an eigenpair of $A$\textnormal{~\cite{brouwer_spectra_2012}}.
\end{lemma}

It follows directly that the spectrum of the quotient graph is contained in the spectrum of the original graph, $\sigma(A_\pi) \subseteq \sigma(G)$~\cite{bunimovich_isospectral_2014}. The inclusion and  ordering of the emergent eigenvalues within the full graph spectrum are guaranteed by Cauchy's Interlacing Theorem~\cite{haemers_interlacing_1995}. This is a key property for our purposes as it guarantees that reducing to the quotient graph does not introduce spurious eigenvalues or discard those relevant to the physics of the system.

The equitable partition defined above requires the constant row-sum condition to hold both within and between cells. It is natural to ask whether the spectral results above survive when these conditions are relaxed. We address this question here for completeness. The almost equitable partitions introduced below relax the equitable partition to require only that connectivity between distinct cells be regular, leaving connectivity within each cell unconstrained. This notion is equivalent to equitable partitions of the graph Laplacian. The results stated above for equitable partitions on adjacency matrices carry over directly to equitable partitions on the graph Laplacian. 

The graph Laplacian is defined as, $L = D - A$, where $D \in \mathbb{R}^{n \times n}$ is a diagonal matrix containing the row sums of the adjacency matrix, and $A$ is the adjacency matrix.

\begin{definition}[Almost equitable partition]
    Let $N(v)$ denote the set of neighbors of $v \in V$. An almost equitable partition is a partition of the vertex set $\pi = \{V_1, \dots, V_k\}$ such that for $V_i, V_j \in \pi$ where $i \neq j$, there exists $d_{ij} \in \mathbb{N}$ such that for all $v \in V_i$, $v$ has exactly $d_{ij}$ neighbors in $C_j$~\cite{timofeyev_cluster_2025}.
\end{definition}

The partition cell a vertex is in determines the number of neighbors that vertex has in each other partition cell. Of note, an almost equitable partition ignores the connectivity within a partition cell. We can recover the equitable partition when the connectivity within each partition is regular.

As similarly done for equitable partitions on the adjacency matrix, we can coarse grain the graph's structure with an almost equitable partition by generating the quotient graph Laplacian. Using Lemma~\ref{lemma: quotient adj} and defining the characteristic matrix of $\pi$ as $P$, the quotient Laplacian matrix of $L$ with respect to the partition $\pi$ is~\cite{timofeyev_cluster_2025,brouwer_spectra_2012}, 
\begin{gather*}
    LP = PL_{\pi} \\
    L_{\pi} = (P^T P)^{-1} P^T L P.
\end{gather*}

This definition uses $\pi$ to partition $L$ into blocks, whose row sums become the entries of $L_{\pi} \in \mathbb{R}^{k \times k}$. The almost equitable partition that generates $L_{\pi}$ induces its own graph that we call the quotient graph, $Q_{L}$.

Applying Lemma~\ref{lemma: eigenpairs} to equitable partitions of the graph Laplacian, the relationship of eigenpairs between the partitioned and non-partitioned graph Laplacian is that the eigenpair $(\lambda, \emph{\textbf{v}})$ of $L_{\pi}$ if and only if $( \lambda, P\emph{\textbf{v}})$ is an eigenpair of $L$~\cite{brouwer_spectra_2012}. 

\section{Minimal Representations} \label{sec:minrep}
\subsection{Minimal representation of QL-bits}\label{subsec:minqlbit}
A natural equitable partition for the QL-bit graph, $\mathcal{G}$, is one where the nodes in subgraphs form the cells, i.e. partitions, $\pi = \{C_1, C_2\}$, and each node in $C_{1}$ has $\ell$ edges in $C_{2}$.  The \textit{k}-regularity of the QL-bit subgraphs  and the $\ell$-regularity of the inter-graph couplings ensure that there are exactly \textit{k} edges from each node within each subgraph and $\ell$ edges from each node in subgraph $A$ to subgraph $B$, respectively. The corresponding quotient graph, $Q^{(q)}$, of the QL-bit consists of two nodes for each subgraph connected by a weighted edge. The characteristic matrix, $S$ that affects the transformation of $\mathcal{R}^{(q)}$ to the adjacency matrix of the quotient graph is a two-column matrix with each column of dimension equal to the number of nodes in each subgraph. The corresponding quotient matrix, $\mathcal{R}_{\pi}(\mathcal{G})$ of $Q^{(q)}$ is given by,
\begin{align}\label{eqn:rpi-qlbit}
    \mathcal{R}^{(q)}_{\pi} = 
    \begin{bmatrix}
        k_{A} & l \\
        l & k_{B}
    \end{bmatrix}
\end{align}
In general, the regularities of each subgraph need not be equal but $k_{A} = k_{B} = k$, in our case. The eigenvalues of $\mathcal{R}^{(q)}_{\pi}$ can be easily computed to be $k\pm l$, are the emergent eigenvalues of $\mathcal{R}^{(q)}$ thus satisfying $\sigma(\mathcal{R}_{\pi}) \subseteq \sigma(\mathcal{R})$. This is also the minimal representation that allows us to recover the emergent spectra of the QL-bit graph and we use this to our advantage while constructing the minimal representation for the Cartesian product of QL-bits.

An equitable partition can generate loops in the quotient graph. In our case, the choice of an equitable partition for the QL-bits will generate loops that manifest as the $k$-regularity of the QL-bit subgraphs (eqn.~\ref{eqn:rpi-qlbit}). The loop constant, $k$ for each maximum equitable partition of a QL-bit, appears as an additive constant along the main diagonal of the adjacency matrix $\mathcal{R}^{(q)}_{\pi}$. Removing the loop constant can be viewed as a global shift of the spectrum by $N_{QL} \cdot k$, where $N_{QL}$ is the number of QL-bits used in the Cartesian product.

In the event the QL-bit subgraphs deviate from $k$- and $\ell$-regular, but still maintain their expander properties, we use Cauchy's Interlacing Theorem to provide bounds and ordering of the eigenstates. How the eigenvalues shift with deviations in $k$- and $\ell$-regularity can be reviewed in Ref~\cite{stewart1990matrix,horn2012matrix}. Let $\mathcal{R}_{\text{approx}}$ be the adjacency matrix of the QL-bit where strict regularity is relaxed. The eigenvalues associated with $\mathcal{R}_{\text{approx}}$ are denoted as $\tilde{\lambda}_i$. Let $\mathcal{R}_{\pi}$ be the $m \times m$ QL-bit quotient matrix. Even when the partition is no longer strictly equitable, the eigenvalues $\mu_i$ of $\mathcal{R}_{\pi}$ interlace those of the full perturbed system, $\mathcal{R}_{\text{approx}}$,
\begin{equation}
    \tilde{\lambda}_i \ge \mu_i \ge \tilde{\lambda}_{N - m + i}.
\end{equation}

The interlacing ensures that the hierarchy of the emergent states of the effective model, in both the perturbed and unperturbed scenarios, are bounded by the emergent states of interest in the full graph.

Similarly, in the event the QL-bit subgraphs deviate from $k$- and $\ell$-regular, but still maintain their expander properties, we use Cauchy's Interlacing Theorem to provide bounds and ordering of the Laplacian eigenstates. Let $L_{\text{approx}}$ be the graph Laplacian of the QL-bit where strict regularity is relaxed. The eigenvalues associated with $L_{\text{approx}}$ are denoted as $\tilde{\nu}_i$. Let $L_{\pi}$ be the $m \times m$ QL-bit quotient Laplacian matrix. Even when the partition is no longer strictly equitable, the eigenvalues $\theta_i$ of $L_{\pi}$ interlace those of the full perturbed system, $L_{\text{approx}}$,
\begin{equation}
    \tilde{\nu}_i \le \theta_i \le \tilde{\nu}_{N - m + i}.
\end{equation}

This interlacing ensures that the hierarchy of the emergent states of the effective model, in both the perturbed and unperturbed scenarios, are bounded by the emergent states of interest in the full graph. In the graph Laplacian picture, the macroscopic dynamics are governed by the smallest non-zero eigenvalues, while the intra-graph expander modes correspond to the largest eigenvalues.

\subsection{Minimal representation of Cartesian product of QL-bits}\label{subsec:mincp}
Cartesian products of several QL-bits produce composite QL-systems whose emergent spectra maps to the tensor product space of the same number of quantum bits. The main idea of this work is to reduce the computational resources required to represent and manipulate such composite QL-systems while preserving their emergent states. The construction of $\mathcal{G}^{(1)} \square \hspace{2pt} \mathcal{G}^{(2)}$, as given in eqn.~\ref{eqn:cartprodt}, generates copies of $\mathcal{G}^{(1)}$ at every node of $\mathcal{G}^{(2)}$, while these copies are correlated according to the associated long-range coupling $\ell^{(2)}$ and local connectivity structure of $\mathcal{G}^{(2)}$. This construction introduces redundancies and therefore raises a natural question of whether that can be exploited to obtain a reduced, yet spectrally equivalent, representation of the full product graph. 

This can be done via two different pathways starting with $N_{QL}$ QL-bit graphs, represented by their adjacency matrices $\mathcal{R}^{(q)}$. In the first one, we create the Cartesian product of the graphs and then define an equitable partition on the Cartesian product to reduce it to its minimal representation. This consideration motivates a specific choice of equitable partition that identifies and collapses the repeated copies of $\mathcal{G}^{(1)}$ at each node of $\mathcal{G}^{(2)}$'s subgraphs, and their associated couplings between each subgraph of $\mathcal{G}^{(2)}$. Additionally, we enforce an equitable partition on $\mathcal{G}^{(1)}$, reducing each of its subgraphs to a node. Figure~\ref{fig: ql to qm}(c) illustrates the equitable partition that reduces the graph copies to a single pair of QL-bits with couplings dictated by that of the second graph in the Cartesian product, and Figure~\ref{fig: ql to qm}(d) illustrates the subsequent equitable partition on $\mathcal{G}^{(1)}$ that results in the minimal representation.

Considering the properties of equitable partitions (as provided in Section~\ref{sec:background}), we state Lemma~\ref{lemma:equitable-partition} establishing that the natural partition on the Cartesian product is equitable, from which our main theorem establishing the minimal representation follows. Theorem~\ref{thm 1} is stated for the product of two QL-bits, but it extends directly to products of an arbitrary number of QL-bit graphs, ($\mathcal{G} = \mathcal{G}^{(1)} \square \hspace{2pt} \mathcal{G}^{(2)} \square \hspace{2pt} ... \hspace{2pt} \square \hspace{2pt} \mathcal{G}^{(n)}$) due to the associativity of the Cartesian product.

 \begin{lemma}(Equitable partition of Cartesian product graph)
Let $\mathcal{G}^{(1)}, \mathcal{G}^{(2)}$ be QL-bit graphs, where each $\mathcal{G}^{(q)}_{A}$ and $\mathcal{G}^{(q)}_{B}$ is $k^{(q)}$-regular and interconnected by an $\ell^{(q)}$-regular coupling, for $q = 1,2$. Let $\mathcal{G} = \mathcal{G}^{(1)} \square \hspace{2pt} \mathcal{G}^{(2)}$, and let $A^{(i)} := V(\mathcal{G}_A^{(i)})$, $B^{(i)} := V(\mathcal{G}_B^{(i)})$ denote the vertex sets of subgraphs $\mathcal{G}_A^{(i)}$ and $\mathcal{G}_B^{(i)}$ of the $i^{\text{th}}$ graph. Define $C_1 = A^{(1)} \times A^{(2)}, \quad
    C_2 = A^{(1)} \times B^{(2)}, \quad
    C_3 = B^{(1)} \times A^{(2)}, \quad
    C_4 = B^{(1)} \times B^{(2)}$. Then $\pi = \{C_1, C_2, C_3, C_4\}$ is an equitable partition of $\mathcal{G}$.
\label{lemma:equitable-partition}
\end{lemma}

\begin{theorem}\label{thm 1}
    Let $\mathcal{G} = \mathcal{G}^{(1)} \square \hspace{2pt} \mathcal{G}^{(2)}$ be the Cartesian product of two graphs, where $\mathcal{G}^{(1)}$ and $\mathcal{G}^{(2)}$ are partitioned in their subgraphs $\{\mathcal{G}^{(1)}_{A}, \mathcal{G}^{(1)}_{B}\}$ and $\{\mathcal{G}^{(2)}_{A}, \mathcal{G}^{(2)}_{B}\}$. We define an equitable partition $\pi = \{C_1, C_2, C_3, C_4\}$ of $\mathcal{G}$ by contracting the subgraphs of $\mathcal{G}^{(1)}$ and $\mathcal{G}^{(2)}$ to nodes, respectively. The resulting quotient graph is the minimal quotient graph that preserves the tensor product topology and emergent eigenpairs, ($\lambda$, S{\textbf{v}}), of the full Cartesian product.
\end{theorem}

\FigQLtoQM

Theorem 1 relies on the equitable partition that results from a single QL-bit graph (as shown in Section~\ref{subsec:minqlbit}) in addition to Lemma~\ref{lemma:equitable-partition} that can be proven from the structure enforced by the Cartesian product of graphs. The proofs for Lemma~\ref{lemma:equitable-partition} and subsequently for Theorem 1 are given in  Appendix~\ref{thm1_proof}.

The other way of finding the minimal representation of the QL product is to first reduce each of the $N_{QL}$ QL-bits to their minimal form using the equitable partition on the QL-bit graph, and then perform the Cartesian product of these minimally represented $N_{QL}$ QL-bits. This uses the fact that the quotienting process commutes with computing the Cartesian product of the $N_{QL}$ QL-bits. For this we state Theorem~\ref{thm:2}. The proof is given in Appendix~\ref{appendix:minrepcp}.
\begin{theorem}\label{thm:2}(Quotienting commutes with the Cartesian product)
Let $G^{(1)}, G^{(2)}$ be graphs with equitable partitions $\pi_1,\pi_2$, characteristic
matrices $S_1,S_2$, and divisor matrices $A_{\pi_1}, A_{\pi_2}$. Then the product partition
$\pi_1\times\pi_2$ is equitable on $G^{(1)}\,\square\, G^{(2)}$, has characteristic matrix
$S_1\otimes S_2$, and has divisor matrix
\[
A_{\pi_1\times\pi_2} = A_{\pi_1}\otimes I + I\otimes A_{\pi_2}.
\]
Equivalently: $\pi\big(G^{(1)}\square G^{(2)}\big) = Q^{(1)}\square Q^{(2)}$--- reducing each
factor and then taking the Cartesian product gives the same quotient graph as taking the
product and then reducing.
\end{theorem} 

The physical reason for this commutativity is the way the Cartesian product of graphs is structured in terms of the two graphs, where one of the graphs is essentially integrated into the other graph. That integration, where graph B replaces each vertex of graph A is equivalent to viewing graph A replacing each vertex of graph B. See Appendix A of Ref.~\cite{scholes_existence_2026} and Appendix B of~\cite{scholes_dynamicsql_2025}. 

Theorem 1, complemented by Theorem 2, is the main result of the paper that allows us to use the Cartesian product picture without having to explicitly compute the full Cartesian product. Figure~\ref{fig:cpquotient} summarizes the two pathways to obtain the minimal representation of the Cartesian product of arbitrary QL-bits, and below we provide an example for the $N_{QL} = 2$ case. Figure~\ref{fig: efficient construction} illustrates the direct construction of a composite QL-system, without first constructing the Cartesian product, for two and three QL-bits. 
\newpage

\begin{figure}
    \centering
    \begin{tikzpicture}[>=Stealth, every node/.style={inner sep=2pt}]
    \tikzset{equation arrow/.style={thick,->,shorten <=2pt,shorten >=5pt}}
    \node (bottom-left)  at (0,0) {$\{\mathcal{Q}^{(1)},\mathcal{Q}^{(2)}, \cdots, \mathcal{Q}^{(N_{QL})}\}$};
    \node (bottom-right) at (8,0) {$\mathcal{G}_{\pi}=\mathop{\Box}\limits_{q=1}^{N_{QL}} \mathcal{Q}^{(q)}$};
    \node (top-left)     at (0,3) {$\{\mathcal{G}^{(1)},\mathcal{G}^{(2)}, \cdots,        \mathcal{G}^{(N_{QL})}\}$};
    \node (top-right)    at (8,3) {$\mathcal{G}=\mathop{\Box}\limits_{q=1}^{N_{QL}} \mathcal{G}^{(q)}$};
    \draw[equation arrow, draw=blue] (bottom-left.east) -- node[above] {$\square$} (bottom-right.west);
    \draw[equation arrow, draw=red] (top-left.east) -- node[above] {$\square$} (top-right.west);
    \draw[equation arrow, draw=blue] (top-left.south) -- node[left]  {$\pi$} (bottom-left.north);
    \draw[equation arrow, draw=red] (top-right.south) -- node[right] {$\pi$} (bottom-right.north);
    \end{tikzpicture}
    \caption{The two pathways of constructing the minimal representation of the Cartesian product of $N_{QL}$ QL-bits made possible by the commutation of the quotienting ($\pi:\mathcal{G}\rightarrow\mathcal{Q}$) with the Cartesian product ($\Box$) formation of the graphs.}
    \label{fig:cpquotient}
\end{figure}
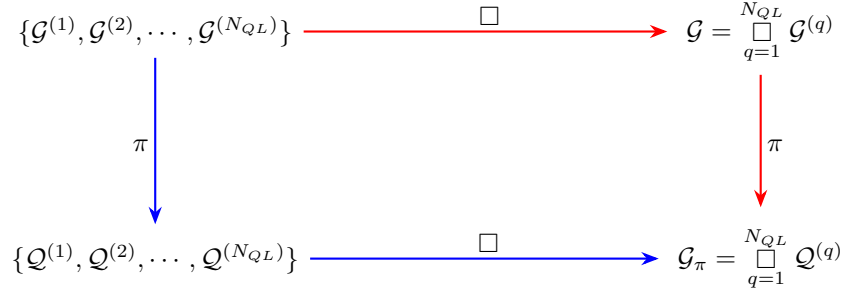

\paragraph*{Example:} We provide here a working example of a two-QL-bit composite graph and its quotient, with their respective computed adjacency matrices. The full Cartesian Product graph for two QL-bits is shown in Figure~\ref{fig: ql to qm}(a) and the corresponding adjacency matrix is given by,
\begin{align}
    \mathcal{R} = 
    \begin{bmatrix}
        A_{A}^{(1)} \otimes I_{n} + I_{n} \otimes A_{A}^{(2)} & I_{n} \otimes C^{(2)} & C^{(1)} \otimes I_{n} & 0 \\
        I_{n} \otimes C^{(2)T} & A_{A}^{(1)} \otimes I_{n} + I_{n} \otimes A_{B}^{(2)} & 0 & C^{(1)} \otimes I_{n} \\
        C^{(1)T} \otimes I_{n} & 0 & A_{B}^{(1)} \otimes I_{n} + I_{n} \otimes A_{A}^{(2)} & I_{n} \otimes C^{(2)} \\
        0 & C^{(1)T} \otimes I_{n} & I_{n} \otimes C^{(2)T} & A_{B}^{(1)} \otimes I_{n} + I_{n} \otimes A_{B}^{(2)}
    \end{bmatrix}
    \label{eqn:full-block-adjacency}
\end{align}
where $A_{A}^{(q)}$ corresponds to subgraph $A$ of the q$^{th}$ QL-bit graph, $\mathcal{G}^{(q)}$. The divisor matrix for the quotient matrix graph following the partitions in Theorem 1 and Lemma 5 can be written in terms of the constant $d_{ij}$'s between the partitions of the full Cartesian product as:
\begin{align}
    \mathcal{R}_{\pi} =
    \begin{pmatrix}
        k^{(1)} + k^{(2)} & \ell^{(2)} & \ell^{(1)} & 0 \\
        \ell^{(2)} & k^{(1)} + k^{(2)} & 0 & \ell^{(1)} \\
        \ell^{(1)} & 0 & k^{(1)} + k^{(2)} & \ell^{(2)} \\
        0 & \ell^{(1)} & \ell^{(2)} & k^{(1)} + k^{(2)}
    \end{pmatrix}.
    \label{eqn:divisor-matrix-4x4}
\end{align}
Here $k^{(q)}$ and $\ell^{(q)}$ correspond to the intra- and inter-graph coupling regularities, and the subgraph regularities for subgraphs $A$ and $B$ for a particular QL-bit is assumed to be equal. The $k$- and $\ell$-regularities of the QL-bits $\mathcal{G}^{(1)}$ and  $\mathcal{G}^{(2)}$ guarantee the existence of exact equitable partitions and ensures that the emergent eigenpairs are preserved under the quotient construction. 

For QL-bit constructions that depart from strict regularity toward more general expander graphs, the framework of almost equitable partitions~\cite{cardoso_aep_2007} provides a natural generalization, wherein the constant neighbor-count condition is relaxed and rigorous bounds on the deviation of the emergent eigenpairs from the exact case can still be established. 

We present here the minimal representation of QL-bits. However, there are intermediate or reduced  representations depending upon the level of abstraction used for each QL-bit. Specifically, the minimal representation depends on which part of the emergent eigenspectrum is physically relevant. If only the isolated emergent states are required, the full reduction of Theorem~\ref{thm 1} applies. If one requires the presence of the incoherent states in the spectrum of the QL-bit, then an intermediate partition like that shown in Figure~\ref{fig: ql to qm} (c) would be a more appropriate reduced representation.
 
\begin{corollary}\label{corollary 1}
    Given the iterative Cartesian product construction of $\mathcal{G}$, the minimal representation of a multi-QL-bit system can be obtained through iterative application of the graph quotient defined in Theorem~\ref{thm 1}. Specifically, the equitable partition and quotient construction can be applied sequentially to any pair of factors at each step, and the associativity of the Cartesian product guarantees that the final minimal quotient graph is independent of the order in which these reductions are performed.
\end{corollary}

The associativity of the graph quotienting process is presented in the language of graph fibrations in Lemma~\ref{lemma: associativity}, where it is shown that sequential fibration maps on disjoint factors of the Cartesian product result in isomorphic quotient graphs regardless of the grouping of factors.

\subsection{Resource reduction in the Cartesian product free construction}\label{subsec:reduction}

In quantum mechanics, composite systems are described by tensor products $\mathcal{H}_{1}\otimes\mathcal{H}_{2}\otimes \cdots \otimes\mathcal{H}_{n}$, and in our framework, this maps to the Cartesian product of graphs $\mathcal{G}^{(1)} \Box \hspace{2pt} \mathcal{G}^{(2)} \Box \cdots \Box \hspace{2pt} \mathcal{G}^{(n)} $.
Constructing coupled QL-bits using the full Cartesian product results in resource scaling of $(N^{N_{QL}})$, where $N_{QL}$ is the number of QL-bits and $N$ is the total number of nodes in a QL-bit, $N = N_{\mathcal{G}^{(q)}_A} + N_{\mathcal{G}^{(q)}_B}$. 

With our minimal representation in mind, we now examine resource scaling of the reduced QL structures. The reduced QL product scales as $D^{N_{QL}}$, for $N_{QL}$ number of $D$-dimensional quantum-like resources, or QL-dits. In this work $D = 2$ for a QL-bit. For the intermediate representation of the Cartesian product of $N_{QL}$ QL-dits, $N_{QL}-1$ QL-dits are reduced to their minimal representation and one full QL-dit is considered. Therefore, the number of nodes in the reduced QL product in this intermediate representation scales as $N(D^{N_{QL}-1})$, where $N$ is the number of nodes in $\mathcal{G}^{(q)}$. The reduction in resource size was first mentioned in~\cite{scholes_existence_2026}, while our work provides the lowest scaling one can expect when constructing QL resources. 

Given our method for obtaining the minimal representation of QL-bit graphs, it is natural to ask what other low-rank approximations and compression methods, such as tensor networks, offer that the equitable partition approach does not. Tensor-network methods, such as matrix product states, compress the representation of a quantum system by exploiting limited entanglement, providing a physically motivated low rank approximations of the studied system Hilbert space. QL systems, by contrast, represent the full Hilbert space in both their full Cartesian product form and their quotient graph form. The basis states defined by the symmetric and anti-symmetric modes of each QL-bit persist under equitable graph partitioning so the full Hilbert space $\mathbb{C}_{1}^2 \otimes \mathbb{C}_{2}^2 \otimes \cdots \otimes \mathbb{C}_{N_{QL}}^2$ remains represented after the reduction. This comes at a cost. Our method incurs an exponential cost in the number of QL-bits whereas tensor networks can achieve sub-exponential scaling in system size when entanglement is limited.

The minimal representation of the Cartesian product resembles the contracted form presented in Appendix A of Ref.~\cite{scholes_existence_2026}, where the ordering of subgraphs in the Cartesian product introduces a partial ordering on the basis, represented as a Hasse diagram for Boolean posets. This structure is further reflected in the adjacency rule of the Cartesian product where two vertices are adjacent if and only if exactly one subgraph label differs, directly mirroring the Hamming distance~\cite{hamming_error_1950} condition $d_{H}=1$.

Contracting the subgraphs of $\mathcal{G}^{(q)}$ to nodes, while similarly reducing the edges between subgraphs of $\mathcal{G}^{(q)}$ to a single edge, maintains the global correlations required to map the emergent states of QL-products to that of the desired state space.

\FigEffConstruct

\FigQLProds

\section{Graph Fibrations and network structure}\label{sec:fibrations}

Following from the method of equitable partitions, one can frame the graph contraction and expansion in the language of graph coverings, fibrations, and liftings. Before introducing the mapping between domains, we introduce the relevant notions of graph covers, fibrations, and liftings~\cite{boldi_fibrations_2002,gray_homotopy_1975,amit_random_2002}. We define a graph as $G$, and its base graph as $B$.

\begin{definition}[Graph Map (Morphism)]
A graph map (or morphism) $\varphi: G \rightarrow B$ from a total graph $G$ to a base graph $B$ consists of a vertex function $\varphi_V: V(G) \rightarrow V(B)$ and an edge function $\varphi_E: E(G) \rightarrow E(B)$ that preserve connectivity. Specifically, if an edge $e \in E(G)$ originates at a source vertex $s(e)$ and terminates at a target vertex $t(e)$, its projection in the base graph must originate and terminate at the projected vertices. This requires $\varphi_V(s(e)) = s(\varphi_E(e))$ and $\varphi_V(t(e)) = t(\varphi_E(e))$.
\end{definition}

\begin{definition}[Graph Fibration]
A map $\varphi: G \rightarrow B$ is a graph fibration if it satisfies a unique lifting property for the terminating ends of edges. Consider any edge $e \in E(B)$ in the base graph that terminates at a vertex $v = t(e)$, and choose any vertex $x \in V(G)$ in the total graph that projects down to $v$ (i.e., $\varphi_V(x) = v$). A fibration guarantees there is exactly one unique edge $\tilde{e}^x \in E(G)$ terminating at $x$ that projects down to $e$.
\end{definition}

\begin{definition}[Graph Opfibration and Covering]
An opfibration between graphs $G$ and $B$ is a map $\varphi: G \rightarrow B$ that satisfies the unique lifting property for the originating ends of edges. For any edge $e \in E(B)$ originating at a vertex $u = s(e)$, and any vertex $x \in V(G)$ projecting down to $u$ ($\varphi_V(x) = u$), there exists exactly one unique edge $^x\tilde{e} \in E(G)$ originating at $x$ that projects down to $e$. This is the oplifting of $e$ at $x$, satisfying $\varphi_E(^x\tilde{e}) = e$ and $s(^x\tilde{e}) = x$. A graph covering is an opfibration in which all fibers $\varphi^{-1}(u)$ have equal cardinality, making the map regular across all vertices of $B$.
\end{definition}

A map that is simultaneously a fibration and an opfibration is called a covering projection. If a covering projection $\varphi: G \rightarrow B$ exists, $G$ is said to be a covering of $B$. The edge $\tilde{a}^x$ of the undirected graph  is called the unique lifting of $a$ at $x$. 

In this framework, the fibration represents the projection from the original system to the equitable partition, while the lifting property provides the mechanism for the inverse operation, reconstructing $G$ from $B$. Figure 4 depicts the equivalence between the quotient graph formed by our equitable partition and the base graph. 

An equitable partition $\pi = \{C_{1}, \dots, C_{k}\}$ of a graph ensures that for all $v \in C_{i}$, the sum of edge weights to cell $C_{j}$ is a constant $d_{ij}$. This ensures that the adjacency sub-matrix $A_{ij}$ has a constant row sum equal to $d_{ij}$. 

This reduction is equivalent to a graph fibration $\varphi: G \rightarrow B$, where the fibers correspond to the cells of the partition. Under this mapping, the preimage $\varphi^{-1}(b_{i}) = C_{i}$ forms a fiber over the vertex $b_{i} \in V(B)$. The inverse of the equitable partition is realized through the lifting of the base graph's adjacency structure: each edge from $b_{j}$ to $b_{i}$ in the base graph $B$ lifts uniquely to exactly $d_{ij}$ edges entering each vertex in the fiber $C_i$ from vertices in fiber $C_j$.

\begin{lemma} \label{lemma: equivalence}
    \emph{(Equivalence of fibrations and covering projections to equitable partitions for undirected graphs)}
    A graph fibration $\varphi:G \to B$ is equivalent to an equitable partition $\pi$ of $V(G)$ where the fibers of the map correspond to the cells $C_i$ of the partition. Furthermore, for undirected regular graphs like QL-bits, this fibration is simultaneously an opfibration, making the map a covering projection~\cite{boldi_fibrations_2002}.
\end{lemma}

The covering projection $\varphi$ defines the many-to-one contraction of $G$ onto the base graph $B$, while the symmetric unique lifting and oplifting properties define the one-to-many expansion, ensuring that the local connectivity seen by any node within a fiber is identical to that from its image in the base graph. Under this equivalence, the adjacency matrix of $G$ admits a block-partition where each block $A_{ij}$ contains the edges between fiber $C_i$ and fiber $C_j$, with constant row sums $d_{ij}$.

\FigBaseGraph

\subsection{Associativity of Fibrations and Coverings}\label{subsec:associative}

The strength of the equivalence between graph fibrations and equitable partitions lies in its composability. For multi QL-bit systems, the iterative construction of the composite graph via the Cartesian product $\mathcal{G} = \mathcal{G}^{(1)} \square \hspace{2pt} \mathcal{G}^{(2)} \square \dots \square \hspace{2pt} \mathcal{G}^{(n)}$ invites an iterative approach to reduction by equitable partitions. The structural relationship between a high-dimensional composite system and its simplified representation is formally governed by the unique lifting property along the fibers. 

Specifically, an equitable partition on a component QL-bit graph $\mathcal{G}^{(q)}$ induces a fibration $\varphi: \mathcal{G}^{(q)} \to B^{(q)}$, where the cells of the partition constitute the fibers of the mapping. Because the individual QL-bit graphs are undirected and regular, this fibration constitutes a graph covering, as the unique lifting property applies symmetrically to both originating and terminating edges. This mapping extends naturally to the composite system, where the reconstruction of the total graph from the base is realized by the unique lifting of arcs/edges from the quotient structure back into the higher-dimensional fibers. In an iterative context, one can apply successive fibrations, where each step collapses local fibers while preserving the equitable property of the global partition. We define this iterative reduction as a sequence of morphisms that maintain the global adjacency structure of the QL-bit system through the algebraic consistency of the quotient graph.

\begin{lemma} \label{lemma: associativity}
    \emph{(Associativity of Fibrations and Coverings)}
    The reduction of a multi-component system by graph fibrations and coverings is associative and independent of the grouping of graph operations. For a composite graph $\mathcal{G} = \mathcal{G}^{(1)} \square \hspace{2pt} \mathcal{G}^{(2)} \square \hspace{2pt} \mathcal{G}^{(3)}$, a fibration $\pi_{12}$ acting on the product $(\mathcal{G}^{(1)} \square \hspace{2pt} \mathcal{G}^{(2)})$ to produce $B^{(12)}$, followed by a fibration $\pi_{(12)3}$ acting on $(B^{(12)} \square \hspace{2pt} \mathcal{G}^{(3)})$, yields a result equivalent to any other grouping of the components. As the fibration of a QL-bit is also an opfibration, and subsequently a covering, we can extend the associativity to coverings as well.
\end{lemma}

For multi-QL-bit systems of arbitrary number of QL-bits, Lemma~\ref{lemma: associativity} is valid because the sequence of reductions is path-independent, it formally guarantees that minimal representations of multi-QL-bit systems can be constructed iteratively, with reduced resource scaling required to instantiate the full Cartesian product. Consequently, regardless of the grouping used during the reduction process, the emergent quantum-like eigenstates and the underlying correlation topology of the composite structure are preserved.

\section{Conclusion and outlook}\label{sec:conclusion}

This work establishes a rigorous framework for reducing the computational and physical resources required to represent composite QL-bit systems. Using equitable partitions of Cartesian products of graphs, we have demonstrated that the full Cartesian product graph of QL-bits admits a minimal structure that preserves emergent eigenspectrum and the correlation structure necessary for quantum-like properties. The minimal construction of the Cartesian product of $N_{QL}$ QL-bits results in a reduction that is exact and brings the resource scaling from $N^{N_{QL}}$ to $2^{N_{QL}}$, where the exponential dependence on the number of QL-bits remains but is now independent of the size of the individual QL-bit graphs. This represents a significant practical improvement, particularly as the number of QL-bits grow.

The connection between equitable partitions and graph fibrations, established through Lemma~\ref{lemma: equivalence}, reveals that in addition to reducing numerical costs, the minimal representation reflects a deeper topological structure of the QL-system. The fibrations perspective shows that the quotient graph preserves the local structural invariants of the full Cartesian product, providing a natural bridge between the graph-theoretic framework of QL-bits and the language of algebraic topology. This connection allows us to look at the QL-systems from a topological lens, specifically homotopy theory and fiber bundle methods.

While the minimal representation provides a theoretical limit to isolate the macroscopic emergent states used for the computational $0$ and $1$ projections, this mathematical limit inherently traces out internal subgraph dynamics. However, in any practical experimental realization, the system is not contracted to idealized singular nodes. Rather, a physical implementation takes the form of the full or the smallest intermediate representation, realized as finite-sized coupled expander graphs.

As a physical implementation of QL-bits will maintain their underlying expander subgraphs, they possess a the necessary spectral gap to support a stable global fixed point. The expander structure at the intermediate scale physically suppresses localized, uncoordinated fluctuations. Consequently, by targeting this intermediate representation, a physical system natively retains the structural robustness against decoherence while simultaneously preserving the dominant emergent states required to define the computational basis. The minimal representation thus serves as the exact macroscopic limit of an inherently robust physical architecture.

Several open questions remain. On the structural side, it is yet to be explored 
whether almost equitable partitions~\cite{cardoso_aep_2007} can yield approximate reductions with controlled spectral error for QL-bits that depart from strict regularity. On the information theoretic side, the following questions require formal treatment: what are the limits on the information that can be encoded in the full Cartesian product versus its minimal quotient graph? Besides the minimal representation, we also discuss intermediate reduced representations of the Cartesian product graph. The question therefore remains as to what level of reduction would allow us to simulate real quantum systems without compromising with the physical aspects of the problem. This, then, naturally motivates further inquiry into the functionality of these graphs, especially if the reduced representation provides us with additional and more meaningful insights into the functions of such networks. In networks corresponding to quotient graph, equitability is both a necessary and a sufficient condition for the existence of cluster-synchronized solutions~\cite{timofeyev_cluster_2025,kovalenko_equitability_2026}. This has been applied to Kuramoto oscillators, Rössler systems, power grids, and neural circuits~\cite{schaub_cluster_2016}.
Finally, the realization of QL-bits in physical systems with initial consideration to networks of mechanical and electrical oscillators~\cite{amati_encoding_2025} remains an important direction, and the minimal representations developed here provide a concrete target for such implementations.

\vspace{1mm}
\emph{Acknowledgments}---We are grateful to Mason Porter and William Munizzi for valuable discussions and insights. G.D.S. acknowledges support from the National Science Foundation under Grant No.\ CHE-2537080. P.N. acknowledges support from the Department of Energy (DOE) Office of Science (SC) under Grant No.\ DOE DE-FOA-0003432 and support from the Gordon and Betty Moore Foundation under Grant No.\ GBMF12976.

%

\widetext
\newpage
\appendix


\section{Minimal Representation of QL Systems (Lemma~3 \& Theorem~1)}\label{thm1_proof}
\subsection{Proof of Lemma~3}\label{appendix:lemma3}
\begin{proof}
A partition $\pi = \{C_1, \dots, C_m\}$ of a graph is \textit{equitable} if, for every pair $i,j$, every vertex in $C_i$ has the same number of neighbors in $C_j$; denote this constant $b_{ij}$. It suffices to fix an arbitrary vertex in each partition, $C_{i}$, and show the neighbor count into every other partition, $C_{j}$, depends only on the subgraph labels, not on the choice of vertex.

The edge rule for the Cartesian product states that $(u_1, u_2) \sim (v_1, v_2)$ in $\mathcal{G}$ if and only if exactly one coordinate is fixed and the other is adjacent in its own factor:
\begin{align}
    (u_1, u_2) \sim (v_1, v_2) \iff
    \big( u_1 = v_1 \ \text{and} \ u_2 \sim_{\mathcal{G}^{(2)}} v_2 \big)
    \ \text{or} \
    \big( u_2 = v_2 \ \text{and} \ u_1 \sim_{\mathcal{G}^{(1)}} v_1 \big).
\end{align}
Here, we use standard graph theory notation for adjacency relations where $u_2 \sim_{\mathcal{G}^{(q)}} v_2$ means $u_{2}$, $v_{2}$ are adjacent within $\mathcal{G}^{(q)}$. Hence every neighbor of $(u_1, u_2)$ arises from moving along exactly one graph while holding the other coordinate fixed, so the neighbor count of $(u_1,u_2)$ splits additively into a $\mathcal{G}^{(1)}$ contribution and a $\mathcal{G}^{(2)}$ contribution. Each contribution depends only on which subgraph of $\mathcal{G}^{(1)}$ (resp.\ $\mathcal{G}^{(2)}$) the corresponding coordinate belongs to, since $\{A^{(q)}, B^{(q)}\}$ is itself equitable in $\mathcal{G}^{(q)}$ by construction (see Section~\ref{subsec:minqlbit}).

Reiterate: $C_1 = A^{(1)} \times A^{(2)}, \quad
    C_2 = A^{(1)} \times B^{(2)}, \quad
    C_3 = B^{(1)} \times A^{(2)}, \quad
    C_4 = B^{(1)} \times B^{(2)}$.
Take $(u_1, u_2) \in C_1$, i.e.\ $u_1 \in A^{(1)}$, $u_2 \in A^{(2)}$, arbitrary and then split the neighboring nodes by which graph moves:

\begin{itemize}
    \item \textit{$\mathcal{G}^{(1)}$ node moves, $\mathcal{G}^{(2)}$ node fixed at $u_2 \in A^{(2)}$:} since $u_1 \in A^{(1)}$ and $\mathcal{G}^{(1)}_A$ is $k^{(1)}$-regular with $\ell^{(1)}$-regular coupling to $B^{(1)}$, there are exactly $k^{(1)}$ neighbors $(v_1, u_2)$ with $v_1 \in A^{(1)}$, landing in $C_1$, and exactly $\ell^{(1)}$ with $v_1 \in B^{(1)}$, landing in $C_3$. These counts are independent of which $u_1 \in A^{(1)}$ was chosen.

    \item \textit{$\mathcal{G}^{(2)}$ node moves, $\mathcal{G}^{(1)}$ node fixed at $u_1 \in A^{(1)}$:} by the same argument applied to $\mathcal{G}^{(2)}$, there are exactly $k^{(2)}$ neighbors landing in $C_1$ and exactly $\ell^{(2)}$ landing in $C_2$, independent of which $u_2 \in A^{(2)}$ was chosen.

    \item \textit{No neighbor lands in $C_4$:} reaching $C_4 = B^{(1)} \times B^{(2)}$ from $C_1$ would require both coordinates to change simultaneously, which the edge rule for $\square$ excludes.
\end{itemize}

Adding the two contributions, the neighbor counts of $(u_1, u_2)$ into $(C_1, C_2, C_3, C_4)$ are
\begin{align}
    \big( k^{(1)} + k^{(2)}, \ \ell^{(2)}, \ \ell^{(1)}, \ 0 \big),
\end{align}
a tuple depending only on $k^{(1)}, k^{(2)}, \ell^{(1)}, \ell^{(2)}$, not on the specific vertex $(u_1, u_2) \in C_1$ chosen. Hence $b_{1j}$ is well-defined for each $j$.

The identical argument, with the roles of $A^{(q)}/B^{(q)}$ exchanged as appropriate, applies verbatim starting from an arbitrary vertex of $C_2$, $C_3$, or $C_4$: in each case one factor contributes its diagonal term $k^{(q)}$ and the other contributes its off-diagonal term $\ell^{(q)}$, with the resulting $b_{ij}$ constant across all vertices of the starting subgraph by the same regularity argument. Thus $b_{ij}$ is well-defined for every pair $i, j \in \{1,2,3,4\}$, which is the equitability condition.
\end{proof}

\subsection{Proof of Theorem 1}
\begin{proof}
We prove that our choice of equitable partition is minimal with respect to preserving the relevant Cartesian product structure and the spectrum.

Let $\mathcal{G} = \mathcal{G}^{(1)} \square \hspace{2pt} \mathcal{G}^{(2)}$ be the Cartesian product of two graphs. Let $Q$ be the quotient graph induced by the equitable partition $\pi = \{C_1, C_2, \dots, C_4\}$ of $\mathcal{G}$ that contracts the individual subgraphs of $\mathcal{G}^{(1)}$ and $\mathcal{G}^{(2)}$ to individual nodes. In the QL-bit case, this gives four total nodes, and divisor matrix that is 4 $\times$ 4.  By construction, $Q = \pi_{12}(\mathcal{G}^{(1)} \Box \hspace{2pt} \mathcal{G}^{(2)})$. 

We must show that any coarser equitable partition $\pi'$ that merges cells of $\pi$ will fail to preserve this product structure or the emergent spectrum. Assume for contradiction that there exists a coarser equitable partition $\pi'$ that merges at least two distinct cells of $\pi$. There are only two cases we must check, as the only two constituents potentially changing are a cell in $\mathcal{G}^{(1)}$ or $\mathcal{G}^{(2)}$. 
\begin{enumerate}
    \item \textbf{Merging across $\mathcal{G}^{(1)}$ subgraphs:} Suppose $\pi'$ merges $C_{1}$ and $C_{2}$. This forces vertices with distinct subgraph affiliations in $\mathcal{G}^{(1)}$ into the same cell. Because the intra-subgraph edges are $k$-regular and inter-subgraph edges are $l$-regular, merging these cells averages the coupling constants. The quotient matrix entries will no longer strictly reflect the isolated $k$ and $l$ degrees, shifting the eigenvalues of the divisor matrix and destroying the specific emergent eigenpairs at $\lambda = k \pm l$.
    \item \textbf{Merging across $\mathcal{G}^{(2)}$ subgraphs:} Suppose $\pi'$ merges $C_{3}$ and $C_{4}$. By identical logic to Case 1, Case 2 follows by the same argument applied to $\mathcal{G}^{(2)}$ by symmetry of the Cartesian product.
\end{enumerate}
Since merging any cells of $\pi$ averages the distinct $k$-regular and $l$-regular degrees required to maintain the emergent eigenspectrum, any such merger results in a contradiction. Additionally, merging any two cells of $\pi$ results in less than four emergent states, which does not recover the emergent states of the two graph Cartesian product case. Therefore, the 4-node quotient graph $Q$ is the minimal valid representation. This argument can be extended to an arbitrary number of graph Cartesian products on $k$-regular graphs, where one considers the pairwise merging for any combination of subgraphs in the higher order Cartesian product. 

Finally, considering the hierarchy of minimal representations, the proof is valid for any level of Cartesian product reduction on the $k$-regular graphs considers. The minimal representation argument is applied to the graphs one wants to reduce, while ignoring the graphs one wants to preserve. 
\end{proof}

\section{Quotienting commutes with the Cartesian product}\label{appendix:minrepcp}
\subsection{Preliminaries for Theorem 2}
\begin{observation}[A Kronecker product of characteristic matrices is a characteristic
matrix]
Let $S_1\in\{0,1\}^{n_1\times r_1}$ and $S_2\in\{0,1\}^{n_2\times r_2}$ be characteristic
matrices, of partitions $\pi_1,\pi_2$. Then $S_1\otimes S_2$ is a characteristic matrix, and
its columns are the indicators of the product cells $C_i^{(1)}\times C_j^{(2)}$ --- i.e.\ it is
the characteristic matrix of the product partition $\pi_1\times\pi_2$.
\end{observation}

\begin{proof}
Index the rows of $S_1\otimes S_2$ by pairs $(u,v)$ with $u\in V_1$, $v\in V_2$, and the
columns by pairs $(i,j)$. By definition of the Kronecker product,
\[
(S_1\otimes S_2)_{(u,v),(i,j)} = (S_1)_{u,i}(S_2)_{v,j}.
\]
Each entry is a product of two numbers in $\{0,1\}$, hence lies in $\{0,1\}$. Row $(u,v)$ has a
single nonzero entry: row $u$ of $S_1$ has its unique $1$ at $i=c_1(u)$ and row $v$ of $S_2$
has its unique $1$ at $j=c_2(v)$, so the product equals $1$ exactly at
$(i,j)=(c_1(u),c_2(v))$ and $0$ otherwise. Finally, column $(i,j)$ is the indicator of
$\{(u,v): u\in C_i^{(1)},\, v\in C_j^{(2)}\} = C_i^{(1)}\times C_j^{(2)}$, which is nonempty
because $C_i^{(1)}$ and $C_j^{(2)}$ are. Hence $S_1\otimes S_2$ has exactly one $1$ per row and
nonempty columns --- a characteristic matrix, of $\pi_1\times\pi_2$.
\end{proof}

\renewcommand{\thelemma}{6}
\begin{lemma}[Converse of the manuscript's Lemma~3, which is also known to hold]
Let $G$ have adjacency matrix $A$, and let $S$ be the characteristic matrix of a partition
$\pi=\{C_1,\dots,C_r\}$. If there exists $B\in\mathbb{R}^{r\times r}$ with $AS=SB$, then $\pi$
is equitable and $B$ is its divisor matrix.
\end{lemma}

\begin{proof}
Fix a vertex $v$ and a cell index $j$, and let $c(v)$ be the index of the cell containing $v$.
Then
\[
(AS)_{v,j} = \sum_w A_{vw}S_{wj} = \sum_{w\in C_j} A_{vw} = |N(v)\cap C_j|,
\]
the number of neighbours of $v$ in cell $C_j$, whereas
\[
(SB)_{v,j} = \sum_k S_{vk}B_{kj} = B_{c(v),j}
\]
since row $v$ of $S$ has its single $1$ in column $c(v)$. The hypothesis $AS=SB$ thus gives,
for every vertex $v$ and cell $C_j$,
\[
|N(v)\cap C_j| = B_{c(v),j}.
\]
The right-hand side depends on $v$ only through its cell $c(v)$; hence any two vertices in the
same cell $C_i$ have equal neighbour-counts $|N(\cdot)\cap C_j| = B_{ij}$ into every cell
$C_j$. That is exactly the definition of an equitable partition, with intercell degrees
$d_{ij}=B_{ij}$ --- i.e.\ $B=A_\pi$. The value is unique: $S$ has full column rank because
$T=S^TS = \mathrm{diag}(|C_1|,\dots,|C_r|)$ is invertible, so $B=(S^TS)^{-1}S^TAS$ is forced.
\end{proof}

\subsection{Proof of Theorem 2}
\begin{proof}
Let $A_1, A_2$ be the adjacency matrices of $G^{(1)}, G^{(2)}$, so the Cartesian-product
adjacency is $A = A_1\otimes I + I\otimes A_2$. Set $S = S_1\otimes S_2$; by Observation~1,
$S$ is the characteristic matrix of $\pi_1\times\pi_2$. Using the mixed-product rule
$(X\otimes Y)(Z\otimes W) = XZ\otimes YW$ together with the per-factor relations
$A_iS_i = S_iA_{\pi_i}$ (the manuscript's Lemma~1 applied to each factor),
\[
AS = (A_1\otimes I + I\otimes A_2)(S_1\otimes S_2)
= A_1S_1\otimes S_2 + S_1\otimes A_2S_2
= S_1A_{\pi_1}\otimes S_2 + S_1\otimes S_2A_{\pi_2}
= SB,
\]
with $B = A_{\pi_1}\otimes I + I\otimes A_{\pi_2}$. Thus $S$ is a characteristic matrix
(Observation~1) satisfying $AS=SB$. By Lemma~6, the partition $\pi_1\times\pi_2$ is equitable
and its divisor matrix is $B = A_{\pi_1}\otimes I + I\otimes A_{\pi_2}$. Since $B$ is precisely
the adjacency matrix of the Cartesian product of the two quotient graphs (whose adjacency
matrices are $A_{\pi_1}, A_{\pi_2}$), this is the claimed identity
$\pi\big(G^{(1)}\square G^{(2)}\big) = Q^{(1)}\square Q^{(2)}$.
\end{proof}

\begin{remark}[Multi-factor corollary]
Applying Theorem~2 inductively to the pair\\
$\big(G^{(1)}\square\cdots\square G^{(N-1)},\, G^{(N)}\big)$ and using only the associativity
of $\square$ gives
\[
\pi\Big(\textstyle\square_{q=1}^{N} G^{(q)}\Big) = \textstyle\square_{q=1}^{N} Q^{(q)},
\]
with divisor equal to the Kronecker sum $\bigoplus_q A_{\pi_q}$ --- the weighted hypercube for
$N$ QL-bits. The manuscript's product-free construction (\S III) and its associativity result
(Lemma~5) are corollaries of this single elementary theorem.
\end{remark}

\section{Equivalence of Fibrations and Covering Projections to Equitable Partitions (Lemma 4)}\label{lemma3_proof}

\begin{proof}
We prove the bidirectional equivalence between a graph fibration $\varphi: G \rightarrow B$ and an equitable partition $\mathcal{P}$ of $V(G)$, and extend this mapping to covering projections for undirected regular graphs.

\textbf{A graph fibration induces an equitable partition.}
Let $\varphi: G \rightarrow B$ be a graph fibration. Define a partition $\mathcal{P} = \{C_i\}$ on $V(G)$ such that each cell $C_i$ is the fiber over a vertex $b_i \in V(B)$; explicitly, $C_i = \varphi^{-1}(b_i)$. 

Let $u \in C_i$ be an arbitrary vertex. We count the number of edges originating in a specific cell $C_j$ and terminating at $u$. By the unique lifting property of a graph fibration, for every edge $e \in E(B)$ originating at $b_j$ and terminating at $b_i$, there exists exactly one lifted edge $\tilde{e}^u \in E(G)$ originating in $C_j$ and terminating at $u$. Therefore, the total number of edges from $C_j$ to $u$ is exactly equal to the number of edges from $b_j$ to $b_i$ in the base graph $B$. Since this quantity depends only on the base graph and is identical for all $u \in C_i$, the partition $\mathcal{P}$ satisfies the constant neighbor count and is therefore equitable.

\textbf{An equitable partition induces a graph fibration.}
Let $\mathcal{P} = \{C_1, \dots, C_m\}$ be an equitable partition of $G$. Let $d_{j,i}$ be the constant number of edges originating in cell $C_j$ and terminating at any specific vertex in cell $C_i$. 

We construct a base graph $B$ with vertices $V(B) = \{b_1, \dots, b_m\}$. For every pair of vertices $(b_j, b_i)$, we add exactly $d_{j,i}$ edges from $b_j$ to $b_i$. We then define a vertex projection map $\varphi_{V}: V(G) \rightarrow V(B)$ such that $\varphi_{V}(v) = b_i$ for all $v \in C_i$, and an edge projection map $\varphi_{E}: E(G) \rightarrow E(B)$ by $\varphi_{E}(u,v) = (b_j,b_i)$ for any edge $(u,v)$ with $u \in C_{j}$ and $v \in C_{i}$. The map $\varphi_{E}$ is well-defined since the equitable condition guarantees that all edges between $C_{j}$ and $C_{i}$ map to the same edge $(b_{j},b_{i})$ in $B$.

We now verify the unique lifting property. Let $e=(b_{j},b_{i}) \in E(B)$ be any edge terminating at $b_{i}$, and let $v \in C_{i}$ be any vertex $\varphi_{V}(v) = b_i$. Since $\mathcal{P}$ is equitable, for any vertex $v \in C_i$, the number of incoming edges from $C_j$ is exactly equal to $d_{j,i}$, each of which project to $e$ under $\varphi_{E}$. Therefore, the base graph $B$ contains exactly $d_{j,i}$ edges from $b_j$ to $b_i$. Consequently, there is a one-to-one correspondence between the incoming edges at $v$ and the incoming edges at $b_i$. This satisfies the unique lifting property, verifying that $\varphi$ is a graph fibration.

\textbf{Extension to Covering Projections:} In the specific case of QL-bits, the graphs $\mathcal{G}^{(q)}$ are undirected. Consequently, the adjacency matrix is symmetric, meaning the constant in-degree $d_{j,i}$ is strictly equal to the constant out-degree $d_{i,j}$. Because the edges are bidirectional, the unique lifting property applies equally to the originating ends of the edges. Thus, the graph fibration $\varphi$ is simultaneously an opfibration. A mapping that is both a fibration and an opfibration is, by definition, a covering projection. Therefore, the equitable partitions on these symmetric QL-bit structures correspond strictly to graph coverings over the base graph $B$.
\end{proof}

\section{Associativity of Fibrations and Coverings (Lemma 5)}\label{lemma4_proof}

\begin{proof}
We want to show that reducing parts of a Cartesian product one step at a time is mathematically valid, and that the associative grouping of these reductions does not change the final quotient structure.

Let $G = G^{(1)} \Box \hspace{2pt} G^{(2)} \Box \hspace{2pt} G^{(3)}$ be the Cartesian product of three graphs. Let $\pi_i: G^{(i)} \to B^{(i)}$ be individual covering projections for each component graph $i \in \{1, 2, 3\}$. We define the composite base graphs as $B^{(12)} = B^{(1)} \Box \hspace{1pt} B^{(2)}$ and $B^{(23)} = B^{(2)} \Box \hspace{1pt} B^{(3)}$.

A fundamental property of covering projections is that the Cartesian product of any two covering projections is itself a covering projection~\cite{dorfler_multiple_1979}. Therefore, combining $\pi_1$ and $\pi_2$ yields a valid covering projection $\pi_{12}: G^{(1)} \Box \hspace{2pt} G^{(2)} \to B^{(12)}$. 

We can extend this to the full three graph system by combining a reduction with an identity map. Thus, the partial reduction map $\Phi_{12}: V(G) \to V(B^{(12)} \Box \hspace{2pt} G^{(3)})$, which reduces the first two graphs and leaves the third unchanged, is inherently a valid covering projection. Similarly, the alternative partial reduction map $\Phi_{23}: V(G) \to V(G^{(1)} \Box \hspace{1pt} B^{(23)})$, which reduces the last two graphs, is also a valid covering projection. No separate manual verification of edge preservation or neighborhood bijection is needed, as it follows directly from the general product theorem.

To prove associativity, we compare the sequence of reductions. Grouping the reductions as $(\pi_1 \times \pi_2)$ followed by $\pi_3$ yields the intermediate graph $B^{(12)} \Box \hspace{2pt} G^{(3)}$ before reaching the final state. Grouping them as $\pi_1$ followed by $(\pi_2 \times \pi_3)$ yields the intermediate graph $G^{(1)} \Box \hspace{1pt} B^{(23)}$. 
    
In both cases, the final composite maps align, mapping any vertex $(u_1, u_2, u_3)$ to the output $(\pi_1(u_1), \pi_2(u_2), \pi_3(u_3))$. Because applying the reductions in either associative grouping equates to the same mathematical function, the reduction operations on Cartesian products are associative.
\end{proof}

\end{document}